%% file: generic_syndrome_decoding_problems.tex
\documentclass[orivec,oribibl]{llncs}

\usepackage{amsfonts}
\usepackage{amssymb}
\usepackage{complexity}
\usepackage{multirow}
\usepackage{array}
\usepackage{scrextend}

\usepackage{amsthm}

\newcommand\numberthis{\addtocounter{equation}{1}\tag{\theequation}}

\input{settings}

\input{macro_llncs}

\newcommand{\ie}{\textit{i.e.}}
\newcommand{\SD}{\textrm{SD}}
\newcommand{\CMSD}{\textrm{CMSD}}

\newcommand{\COMMENT}[1]{}
\renewcommand{\aa}{\mathcal{A}}

\renewcommand{\wt}[1]{{wt}\left(#1\right)}

\usepackage{algorithm}
\usepackage[noend]{algpseudocode}

\newcommand{\I}{\mathcal{I}}
\newcommand{\zerov}{\mathbf{0}}

\title{Classical and Quantum algorithms for generic Syndrome Decoding problems and applications to the Lee metric}

\author{Andr\'e Chailloux \inst{1} \and Thomas Debris-Alazard\inst{2} \and Simona Etinski \inst{1}}
\institute{Inria de Paris, EPI COSMIQ \and Inria Saclay\\
	\email{\{andre.chailloux, thomas.debris, simona.etinski\}@inria.fr}}

\begin{document}
	
\maketitle 
		
{\abstract{The security of code-based cryptography usually relies on the hardness of the syndrome decoding (SD) problem for the Hamming weight. The best generic algorithms are all improvements of an old algorithm by Prange, and they are known under the name of Information Set Decoding (ISD) algorithms. This work aims to extend ISD algorithms' scope by changing the underlying weight function and alphabet size of SD. More precisely, we show how to use Wagner's algorithm in the ISD framework to solve SD for a wide range of weight functions. We also calculate the asymptotic complexities of ISD algorithms both for the classical and quantum case. We then apply our results to the Lee metric, which currently receives a significant amount of attention. By providing the parameters of SD for which decoding in the Lee weight seems to be the hardest, our study could have several applications for designing code-based cryptosystems and their security analysis, especially against quantum adversaries.}}
	
\section{Introduction}	
	
Code-based cryptography is one of the leading proposals for post-quantum cryptography,
and it traditionally relies on the hardness of the syndrome decoding problem. For fixed $q, n,k,w$, the problem is defined as follows: starting from a parity check matrix $\Hm \in \F_q^{n\times(n-k)}$, and a syndrome $\sv \in \F_q^{n-k}$, the goal is to find a vector $\ev \in \F_q^n$ such that $\Hm \ev = \sv$, and $\ev$ has the Hamming weight\footnote{The Hamming weight of a vector $\ev = (e_1,\dots,e_n)$ is $|\ev|_{\textup{H}} \eqdef |\{i : e_i \neq 0\}|$.} $w$. This problem has been studied for a long time, and mostly for the alphabet size $q = 2$. Despite many efforts, the best algorithms for solving this problem \cite{P62,S88,D91,B97b,MMT11,BJMM12,MO15}  require an exponential running time, and they are all refinements of the original Prange's algorithm \cite{P62}. As such, they are all commonly known under one name: Information Set Decoding (ISD) algorithms.

It  is, however, notoriously difficult to put the syndrome decoding problem into practice. For example, constructing an efficient signature scheme in code-based cryptography often requires utilizing pseudo-random functions, and some other cryptographic assumptions. A generalized version of the problem promises to be harder and to offer
a more exploitable structure that leads to creating more efficient constructions. Like DURANDAL [ABG+19], some proposals replace the Hamming weight with the rank metric based weight, which allows designers to use a Schnorr-Lyubashevski type signature. Another proposal,  $\mathsf{WAVE}$ signature scheme \cite{DST19a}, utilizes syndrome decoding problem for which $q = 3$, and the Hamming weight is large. This further enables one to construct and exploit a trapdoor one-way preimage sampleable function, which would not be possible for $q = 2$ or $q = 3$ in small weight.
	
These examples already show the usefulness of going beyond $q=2$ and the Hamming weight setting. We are, however, still at an early stage of using these variants for cryptographic schemes. Therefore, it is important to study their hardness, especially against quantum computers, since a big appeal of code-based cryptography is post-quantum security.
		
\paragraph{Our work.} In this paper, we perform a generic analysis of different ISD algorithms. The analysis is applicable to any weight function $wt : \F_{q}^{n} \rightarrow \mathbb{R}_+$ satisfying $wt(\ev) \eqdef \sum_{i = 1}^n wt’(e_i)$, for some function $wt’ : \F_q \rightarrow \mathbb{R}_+$, and $wt’(0) = 0.$ However, we primarily focus on the Lee weight analysis, and the comparison between the Lee and Hamming weight. The reason we chose these two weight functions is that the two are commonly encountered in coding theory, and both led to proposals for cryptographic schemes.

Which ISD algorithms do we study here? We analyze algorithms by Prange and Stern/Dumer and the ISD algorithms based on Wagner's approach to solving a Generalized Birthday problem \cite{Wag02}. Starting from \cite{BCDL19}, where classical algorithms for a ternary alphabet and the Hamming weight were analyzed, we broaden the analysis to the higher alphabet sizes, usage of a different weight function, and the study of both classical and quantum algorithms. This is the first time such a generic analysis of quantum ISD algorithms was done since the work of \cite{KT17a} that studied only the standard case of $q = 2$ and the Hamming weight.

To perform such a generic analysis, we need a way of computing sphere surface areas in a vector space endowed with an arbitrary metric. More precisely, we aim to calculate the sizes of  sets of the form $\{\ev \in \F_q^n : wt(\ev) = p\}$. To do this, we start with the approach presented in \cite{A84a}, applied to the Lee metric case, and we derive a convex optimization method for calculating the asymptotic sphere surface area independently of the metric. We thus provided a simple approach to analyzing syndrome decoding problems in a vector space endowed with an arbitrary metric and a weight function derived from it. 

Our framework can also be used for studying the security of the Restricted Syndrome Decoding problem \cite{BBCHPSW2020}. Nevertheless, it does not work for the rank metric norm where we do not know how to construct ISD algorithms better than Prange’s algorithm\footnote{There are other algorithms \cite{BBBGNRT20,BBCGPSTV20} based on Gr\"obner basis that perform better than ISD algorithms for the rank metric.}.
	
\section*{Notations}

Throughout the paper, we use $[n] \eqdef \{1,\dots,n\}$ and, given a finite set $\mathcal{E}$, we denote by $|\mathcal{E}|$ its size. We consider a weight function $wt : \F_{q}^{n} \rightarrow \mathbb{R}_+$ which satisfies the following:
	
\begin{multline}\label{Constriant:WeightConstraint}
	\exists wt' : \F_q \rightarrow \mathbb{R}_+ \mbox{ : }  wt'(0) = 0 \mbox{ and } \forall \ev = (e_1,\dots,e_n) \in \F_q^n, \wt{\ev} = \sum_i wt'(e_i).
\end{multline}

This weight function is usually - but not always - obtained as $wt(\vec{x}) = d(\xv,0)$ where $d$ is a distance. We will sometimes use the terminology of distance instead of weight when this is the case. When $q$ and $wt$ are fixed and explicit, we define the surface area  of a sphere of weight $w$ in a vector space of dimension $n$ as:
$$
S^n_w \eqdef \left| \{\ev \in \F_q^n : \wt{\ev} = w\} \right|. $$ 

\section{Quantum preliminaries}\label{Section:AppendixQuantumPreliminaries} 

We refer to \cite{NC00} for a basic introduction to quantum computing. 
In this paper, we use the canonical gate model where the running time of a quantum algorithm is the number of gates in its corresponding circuit description. We utilize the QRAM model, for which we assume the operation $U_{QRAM} : \ket{i}\ket{y}\ket{b_1,\dots,b_n} \rightarrow \ket{i}\ket{y + x_i}\ket{b_1,\dots,b_n}$ can be done in time $polylog(n)$ when each $b_i$ is a single bit.

\paragraph{Grover's algorithm. \cite{Gro96}} For a function $f : \zo^n\rightarrow \zo$ that has an efficient classical description, Grover's algorithm can find $x$ such $f(x) = 1$ in time $O(poly(n) 2^{n/2})$ if such an $x$ exists and output 'no solution' otherwise.

\paragraph{Amplitude amplification. \cite{BH97}} Fix a function $f : \zo^n\rightarrow \zo$ that has an efficient classical description. Consider then a quantum algorithm $\aa$ that outputs $x$ such that $f(x) = 1$ with probability $p$ and does not perform intermediate quantum measurements. Using amplitude amplification, one can find $x$ such that $f(x) = 1$ by making $O(\frac{1}{\sqrt{p}})$ calls to $\aa$. Notice that if we start from a classical algorithm $\aa$, there are generic ways to run $\aa$ coherently as a quantum algorithm $\aa'$ that does not have intermediate quantum measurements and behaves exactly like $\aa$.
	
\section{Syndrome Decoding Problems}\label{Section:SyndromeDecodingProblems}

When we fix an alphabet size $q$ and a weight function $wt$, the syndrome decoding problem is defined as follows:
	
\begin{problem} Syndrome Decoding $\SD(n,k \le n,w)$
	\begin{itemize}
		\item Input: A matrix $\Hm \in \F_q^{(n-k)\times n}$, a column vector (the syndrome) $\vec{s} \in \F_q^{n-k}$. 
		\item Goal: Find a column vector $\vec{e} \in \F_q^n$ s.t. $\Hm \ev = \sv$ and $\wt{\ev} = w$.
	\end{itemize}
\end{problem}
	
The decision version of this problem, which asks whether there exists a vector $\ev$ of weight $w$ such that $\Hm \ev = \sv$, is $\NP$-complete for $q = 2$ with the Hamming weight function \cite{BMT78}. 
	
Consider now the input distribution $\mathcal{D}$ sampled as follows: pick a random matrix $\Hm \in \F_q^{(n-k) \times n}$ of  rank $n-k$, pick a random $\ev \in \F_q^n$ with $wt(\ev) = w$, and output $(\Hm,\sv = \Hm \ev)$. Notice that the problem always has at least one solution for this distribution and that $\SD$ is believed to be hard, even against quantum computers. That is why, in this paper, we study algorithms for $\SD$ with this input distribution. We only consider a prime $q$ to avoid attacks that would use sub-fields of the alphabet field $\F_q$.
	
Another problem of interest, which we call Checkable Multiple Syndrome Decoding, is the following: 
\begin{problem} Checkable Multiple Syndrome Decoding $\CMSD(n, m, w, Y, Z)$
	\begin{itemize}
		\item Input: A matrix $\Hm \in \F_q^{m \times n}$, a syndrome $\vec{s} \in \F_q^{m}$.
		\item Goal: output the description of a function $f : [Y] \rightarrow \F_q^n$ such that $f$ is efficiently computable, and $|\{\ev : \ev \in \textup{Im}(f), \Hm \ev = \sv \text{ and  }\wt{\ev} = w\}| = Z$.
	\end{itemize}
\end{problem}
	
This problem is a bit funny looking at first sight, but we are interested in it because, in our framework, it is used as a building block for solving the generic $\SD$ problem. It is very similar to asking for $Z$ solutions to the syndrome decoding problem. Indeed, from a description $f$, one can output $Z$ solutions to $\SD$ in time $Y$ by enumerating all the $f(1),\dots,f(Y)$. Reciprocally, if one can find $Z$ solutions $\ev_1,\dots,\ev_Z$ to $\SD(n,m,w)$ in time $T \ge Z$, then one can solve $\CMSD(n, m, w, Y, Z)$ by defining $f(i) = \ev_i$.
	
In the quantum setting, we want to have access to the function $f$ but without paying for a time cost of $Z$ for writing down these solutions. That will allow us to search over solutions more efficiently, using Grover's algorithm, and also justifies the slightly odd definition. Another remark is that while $f$ should be efficiently computable, it need not have an efficient description. Typically, $f$ can store some large precomputed databases, but computing $f(x)$ will only query the database a small number of times. 
	
\section{Information Set Decoding Algorithms for any Metric}

We present Information Set Decoding algorithms for $\SD$, which consist of a partial Gaussian elimination followed by solving an instance of $\CMSD$. The description here is essentially the one from \cite{BCDL19} with the difference that here we use the $\CMSD$ problem.
	
\subsection{Information Set Decoding Framework} \label{Section:InformationSetDecodingFramework}

Fix $\Hm\in\F_{q}^{(n-k)\times n}$ of rank $(n-k)$ and $\sv\in\F_{q}^{n-k}$. Recall that we want to find $\ev\in\F_{q}^n$ such that $\wt{\ev}= w$ and $\Hm\ev=\sv$. Let us introduce $\ell, p, Y,$ and $Z$, four parameters of the system that we consider fixed for now. In this framework, an algorithm for solving $\SD(n,k,w)$ consists of $4$ steps: a permutation step, a partial Gaussian Elimination step, a $\CMSD$ step, and a test step.
	
\begin{enumerate}
	\item \emph{Permutation step.} Pick a random permutation $\pi$. Let $\Hmpi$ be the matrix $\Hm$ with the columns permuted according to $\pi$. We now want to solve $\SD(n,k,w)$ on inputs $\Hmpi$ and $\sv$.
		
	\item \emph{Partial Gaussian Elimination step.} If the top left square submatrix of $\Hmpi$ of size $n-k-\ell$ is not of full rank, go back to step $1$ and choose another random permutation $\pi$. That happens with constant probability.\footnote{For $q = 2$, this happens with probability at least $0.288$ and this probability increases as $q$ increases (see \cite{Coo00}, for example).} If the submatrix is of full rank, perform Gaussian elimination on the rows of $\Hmpi$ using the first $n-k-\ell$ columns. Let now $\Sm \in \F_q^{(n-k)\times (n-k)}$ be the invertible matrix corresponding to this operation. There are two matrices then, $\Hm' \in \F_q^{(n-k-\ell) \times (k+\ell)}$ and $\Hm'' \in \F_q^{\ell \times (k+\ell)}$, such that:
		
	\[
	\Sm\Hmpi = 
	\begin{pmatrix} 
		\un_{n-k-\ell} & \Hm' \\ 
		\mathbf{0} & \Hm''
	\end{pmatrix}.
	\]
		
	A vector $\ev \in \F_q^n$ can be written as $\ev = \begin{pmatrix} \ev' \\ \ev'' \end{pmatrix}$, where $\ev' \in \F_q^{n-k-\ell}$ and $\ev'' \in \F_q^{k+\ell}$, and one can write ${\Sm}\sv = \begin{pmatrix} \sv' \\ \sv'' \end{pmatrix}$, with $\sv' \in \F_q^{n-k-\ell}$ and $\sv'' \in \F_q^{\ell}$.
		
    \begin{align*} 
    	\Hmpi{\ev} = {\sv} 
    	& \iff \Sm\Hm_\pi {\ev} = \Sm {\sv} \\
    	& \iff 
    	\begin{pmatrix} 
    		\un_{n-k-\ell} & \Hm' \\ 
    		\mathbf{0} & \Hm''
    	\end{pmatrix} 		
    	\begin{pmatrix} 
    		{\ev'} \\
    		{\ev''} 
    	\end{pmatrix} 
    	= 
    	\begin{pmatrix} 
    		{\sv'} \\
    		{\sv''} 
    	\end{pmatrix} \\ 
    	&\iff 
    	\left\{
    	\begin{array}{ll} 
    		{\ev'} + \Hm' {\ev''} = {\sv'} \\ 
    		\Hm''{\ev''} = {\sv''}
    	\end{array} 
    	\right. 
    	\numberthis 
    	\label{EQ1}
    \end{align*} 
		
	To solve the problem, we try to find a solution $\begin{pmatrix} \ev' \\ \ev'' \end{pmatrix}$ to the above system such that $\wt{\ev''} = p$ and $\wt{\ev'} = w-p$.
		
	\item \emph{The $\textup{\CMSD}$ step.} Solve $\CMSD(k+\ell,\ell,p,Y,Z)$ on input $(\Hm'',\sv'')$, and let $f$ be the output function. 
		
	\item \emph{The test step.} For each $i \in [Y]$, let $\ev''_i =f(i)$ and let $ {\ev'_i} = {\sv'} - \Hm'{\ev''_i}$. For each $i$ such that $\Hm''\ev'' = \sv''$, Equation \eqref{EQ1} ensures that $\Hmpi \begin{pmatrix} \ev'_i\\ \ev''_i \end{pmatrix} = {\sv}$.  If $\wt{\ev''_i} = p$ and $\wt{\ev'_i} = w-p$, $\ev_i = \begin{pmatrix} \ev'_i\\ \ev''_i \end{pmatrix}$ is therefore a solution to $\SD(n,k,w)$ on inputs $\Hmpi$ and $\sv$. The solution to $\SD(n,k,w)$ can then be turned into a solution of the initial problem by permuting the indices, as detailed in Equation \eqref{eq:permutation} below. If we do not find any solution after checking all $i \in [Y]$, we go back to step $1$. 

\end{enumerate} 
	
At the end of the protocol, we have a vector $\ev$ such that $\Hmpi \ev = \sv$ and $\wt{\ev} = w$. Let $\ev_{\pi^{-1}}$ be the vector $\ev$ with the permuted coordinates according to $\pi^{-1}$. Hence,

\begin{equation}
	\Hm {\ev}_{\pi^{-1}} = \Hmpi {\ev} = {\sv} 
	\quad \textrm{ and } \quad 
	\wt{\ev_{\pi^{-1}}} = \wt{\ev} =  w.
	\label{eq:permutation}
\end{equation}
Therefore, $\ev_{\pi^{-1}}$ is a solution to the problem.
	
\subsection{Information Set Decoding: Complexity Analysis (Classical and Quantum)}

We fix $q$ and a weight function $wt$. Recall that for any $n$ and $w$, the surface area of a sphere (according to $wt$) of radius $w$ in $\Fq^{n}$ is defined as:	
$$S^{n}_w = |\{\ev \in \F_q^n: \wt{\ev} = w\}|.$$

With this definition at hand, we now present the complexity analysis of the algorithm for solving $\SD(n,k,w)$ for fixed parameters $\ell,p,Y,Z$ (see section \ref{Section:InformationSetDecodingFramework} for more details).
	
\begin{lemma}\label{Lemma:Probability1}
	Let $P_1$ be the probability that at step $4$, for a fixed $i$, $\wt{\ev'_i} = w - p$. We have: 
$$P_1 = \min\{1,O(\frac{S^{n-k-\ell}_{w-p}}{\max\{1,\min\{S^n_w q^{-\ell},q^{n-k-\ell}\}\}})\}.$$
\end{lemma}

This lemma can be seen as a generalization of Proposition 2 of \cite{BCDL19} (where a max was omitted) for any weight function. 

\begin{proof}
Let $S = \{\ev : \wt{\ev} = w \wedge \Hm_\pi \ev = \sv\}$ be the set of solutions to our syndrome decoding problem on input $\Hm_\pi,\sv$. Let also $S_2 = \{\ev = \begin{pmatrix} \ev' \\ \ev'' \end{pmatrix} : \wt{\ev} = w \wedge \Hm'' \ev'' = \sv''\}$, where $\Hm''$ is the matrix from step $2$. By definition, $S \subseteq S_2$, so we have that $S$ has average size $\max\{1,S^n_wq^{-(n-k)}\}$ and $S_2$ has average size $\max\{S^n_w q^{-\ell},1\}.$
	
Fix $i$ and $\ev''_i = f(i)$ satisfying $\Hm'' \ev''_i = \sv''$ and $\wt{\ev''_i} = p$. $T_i =  \left\{\ev_i = \begin{pmatrix} \ev'_i \\ \ev''_i \end{pmatrix} : \wt{\ev_i} = w\right\}$. $T_i$ is of average size $S^{n-k-\ell}_{w-p}$. Step $4$ will find a solution if $T_i \cap S \neq \emptyset$. Since $T_i \subseteq S_2$ and is uniformly distributed in this set, this happens with the following probability: 
	
\begin{align*}
	P_1 & = \min\{1,O(\frac{|T_i||S|}{|S_2|})\} =  \min\{1,O(\frac{S^{n-k-\ell}_{w-p} \cdot \max\{1,S^n_w q^{-(n-k)}\}}{\max\{S^n_w q^{-\ell},1\}})\} \\
	&  = \min\{1,O(\frac{S^{n-k-\ell}_{w-p}}{\max\{1,\min\{S^n_w q^{-\ell},q^{n-k-\ell}\}\}})\}.
\end{align*}

\end{proof}

We now present our generic formula for the running time of the Information Set Decoding algorithm from Section \ref{Section:InformationSetDecodingFramework}. 

\begin{proposition}\label{Proposition:ISDC}
	Fix parameters $\ell,p,Y$, and $Z$ of the information set decoding algorithm. The classical running time of the algorithm, $T_{\textup{ISD}}$, is given as: 
$$ T_{\textup{ISD}} = O\left(\max\left\{1,\frac{1}{P_1 Z}\right\} \cdot\left(\textup{poly}(n) + T_{\textup{CMSD}} + \textup{poly}(n) Y \right)\right),$$
where $P_1$ is the probability from the above lemma, and $T_{\textup{CMSD}}$ is the running time of step $3$, {\ie }, the time required for solving $\textup{\CMSD}(k+\ell,\ell,p,Y,Z)$.
\end{proposition}

\begin{proof}
	Steps $1$ and $2$ take time $\textup{poly}(n)$, step $3$ takes time $T_{\textup{CMSD}}$, and step $4$ takes time $\textup{poly}(n)$ for each $i \in [Y]$, hence the right part of the expression. How many times does the algorithm loop over this process? Step $2$ succeeds with constant probability, and step $4$ finds a solution with probability $1 - (1- P_1)^Z$, so it loops over the steps $O\left(\frac{1}{1 - (1- P_1)^Z}\right) =  O\left(\max\left\{1,\frac{1}{P_1Z}\right\}\right)$ times, hence the result.
\end{proof}

\paragraph{The quantum setting.} Our formulation allows for a simple extension to the quantum setting. We consider the algorithm described earlier with the following two changes: $(1)$ in step $4$, the algorithm uses Grover's search to check whether there is $i$ such that $f(i)$ gives us a solution; (2) for each loop, {\ie }, each time the algorithm starts from step $1$, it finds a solution with probability $p = \Omega(\min\{1,P_1 Z\})$. This loop can be made coherently with a quantum algorithm $\aa$ that does not do intermediate measurements and outputs a solution with probability $p$.  The algorithm then use amplitude amplification to find a solution by repeating the loop $O(\frac{1}{\sqrt{p}})$ times. 
	
\begin{proposition}\label{Proposition:ISDQ}
	Fix parameters $\ell,p,Y,$ and $Z$ of the information set decoding algorithm. The quantum running time of the algorithm, $T^Q_{\textup{ISD}}$, is given as: 
$$ T^Q_{\textup{ISD}} = O\left(\sqrt{\max\left\{\frac{1}{Z P_1},1\right\}} \cdot \left(\textup{poly}(n) + T_{\textup{CMSD}} + \textup{poly}(n)\sqrt{Y}\right)\right),$$
where $P_1$ is the probability from Lemma \ref{Lemma:Probability1}, and $T_{\textup{CMSD}}$ is the running time of step $3$, {\ie }, of solving $\textup{\CMSD}(k+\ell,\ell,p,Y,Z)$.
\end{proposition}

\begin{proof}
	Again, Steps $1$ and $2$ take time $\textup{poly}(n)$, and step $3$ takes time $T_{\textup{CMSD}}$. In step $4$, the algorithm runs Grover's search, so this whole step takes time $\textup{poly}(n)O(\sqrt{Y})$. That can be done because the function on input $i$ determines whether $\wt{\ev'_i} = w-p$ runs in polynomial time (since $f$ runs in polynomial time). As we described above, we repeat the loop $O\left(\sqrt{\max\left\{\frac{1}{Z P_1},1\right\}}\right)$ times, which gives the result.
\end{proof}

\paragraph{The full ISD algorithm.} To find the best ISD algorithm for solving $\SD(n,k,w)$, we minimize the running time of the algorithm presented earlier over parameters $p,\ell,Y,$ and $Z$. In many cases, we do not have full control over $Y$ and $Z$, which are predetermined from other values. For instance, in Wagner's algorithm, we present next, there is an extra parameter $a$ (the number of levels) that predetermines $Y$ and $Z$, so we optimize over $p,\ell$, and $a$. 
	
\section{Solving $\textup{\CMSD}$}

This section presents our analysis of the application of Wagner's algorithm \cite{Wag02} to solving $\CMSD(N,m_0N,\omega_{0}N,Y,Z)$\footnote{As Wagner's algorithm is used for solving Generalized Birthday Problem, it can be easily seen that is well suited for solving $\CMSD$ problems, too.}. We first present the list merging procedure, which we utilize throughout the section, and then the two versions of our algorithm: the first one that aims to solve the $\CMSD$ problem using classical algorithms only, and the second one that utilizes both classical and quantum algorithms.
	
Notice here the change of the variables' names when referring to the $\CMSD$ problem. It is introduced so that our statements can be made independently of the previous section. Notice also that the asymptotic values of the algorithms' running times are calculated when $N$ goes to $+\infty$ and that when presenting a proof, we ignore all the polynomial and constant terms. 

\subsection{List Merging}\label{Section:ListMerge}

Let us take $3$ lists of vectors in $\F_q^n$: $L_1, L_2$, and $L$. Take also a set $J \subseteq [n]$ and a random vector $\tv \in \F_q^{|J|}$. The merging of $L_1$ and $L_2$ into $L$ is done using the following algorithm:

\paragraph{List merging algorithm.}
\begin{itemize}
	\item Start from an empty list $L$, and sort the elements of $L_1$ according to the lexicographic order on the $J$ coordinates.
	\item For each vector $\yv \in \F_q^n$ from the list $L_2$, search for elements $\xv \in \F_q^n$ of $L_1$ that satisfy: $\xv_{|J} = \yv_{|J} + \tv_{|J}$, where $\xv_{|J} \eqdef (x_{j})_{j \in J}$, $\yv_{|J} \eqdef (y_{j})_{j \in J}$, and $\tv_{|J} \eqdef (t_{j})_{j \in J}$. For each solution found, add $\xv + \yv$ in $L$ and register the references to $\xv$ and $\yv$.
\end{itemize}

\paragraph{Running time.}

Sorting $L_1$ on $J$ coordinates is done in time $O(\log(|L_1|))$ using dichotomic search. If there are $s_\yv$ solutions for a fixed $\yv$, the algorithm takes $O(s_\yv\log(|L_1|))$ time to find them, and the total size of $L$ is $\sum_{\yv} s_\yv$. Therefore, the algorithm takes time $\widetilde{O}\left(|L_1|\right)$ for the first step, {\ie }, to sort $L_1$, and it takes $\widetilde{O}\left(\max\{|L_2|,\sum_i s_\yv \}\right)$ for the second step. Overall, the algorithm takes time $\widetilde{O}(\max\{|L_1|, |L_2|,|L|\})$.

\paragraph{Expected number of solutions.} If the elements in $L_1$ and $L_2$ are random vectors in $\F_q^n$, there is, on average, $|L| = \frac{|L_1||L_2|}{q^{|J|}}$ elements in the merged list.

\paragraph {List merging operator.} To enable a succinct representation of this procedure in the rest of the text, we define the list merge operator on a set $J$ and random vector $\tv$, denoted as $\bowtie_{J}^\tv$:
$$L = L_1 \bowtie_{J}^\tv L_2 = \{\xv + \yv: \xv \in L_1, \yv \in L_2, \xv_{|J} + \yv_{|J} = \tv_{|J}\}.$$
		
\subsection{First Variant} \label{Section:WagnerVar1}

We present here an approach to solving the $\CMSD$ problem, based on Wagner’s algorithm \cite{Wag02}, which utilizes classical algorithms only and is closely related to the original Wagner's algorithm.\\

We start from relevant definitions. For a number of levels $a$, where $2^a|n$, and for each $i \in [2^a]$, we define:
\begin{multline*}
	\I_i  \eqdef \{\bv \in \F_q^n : \bv = (\zerov^{(i-1) n/2^a},\bv_i,\zerov^{(2^a-i)  n/2^a}) \textrm{ with } \bv_i \in \F_q^{n/2^a} \wedge \wt{\bv_i} = N\omega_0/2^a\},
\end{multline*}
\begin{align*}
	L^f_i &\eqdef  \{\Hm \cdot \bv\}_{\bv \in \I_i},\footnotemark
\end{align*}\footnotetext{From the definitions, it can be easily seen that $|L^f_i| = \frac{s_{\omega_0}}{2^a}$.}where $\Hm$ is a parity check matrix, defined in Section \ref{Section:InformationSetDecodingFramework}.

The sets used for the indexing the lists in the merging procedure (as described in Section \ref{Section:ListMerge}) are chosen so that they form a partition of $[n]$, $\ie$:
$$\forall j,j' \in [a], \quad J_j \subseteq [n], \quad \bigcup_j J_j = [n], \quad J_j \cap J_{j'} = \emptyset, \text{ when } j \neq j'.$$

The random vectors (again, described in Section \ref{Section:ListMerge}) are chosen such that they satisfy the following constraint:
$$\forall i \in [2^a], \quad \forall j \in [a], \quad \tv_j^i \in \F_q^n , \quad \sum_{i} (\tv_j^i)_{|J_j} = \sv_{|J_j},$$
where $\sv_{|J_j}$ refers to the syndrome, from Section \ref{Section:InformationSetDecodingFramework}, indexed by $J_j$.


\paragraph{List creation and merging.} The algorithm starts by constructing $2^a$ lists of the same sizes: $L_i \subseteq L^f_i$, for all $i \in [2^a]$\footnote{There are previous description where  $L_i  = L^f_i$, but the inclusion improves the algorithm efficiency.}. At each algorithm's level, the lists are then taken by pairs, $\{L_{2i - 1},L_{2i}\}$, and merged using the list merging procedure described in the previous subsection. More precisely, at the first level, the pairs are merged on a set $J_1$ and a random vector ${\tv_1^i}$ ($\ie$, $\bowtie_{J_1}^{\tv_1^i}$ is performed). From the $2^{a-1}$ created lists, at the second level, pairs are taken again and merged similarly using the operator $\bowtie_{J_2}^{\tv_2^i}$, for each $i \in [2^{a-2}]$. The same procedure continues up to the top level, where only $2$ lists remain and the list merging is performed using $\bowtie_{J_a}^{\tv_a^1 = \sv_{|J_a}}$. A function $f$, required for the $\CMSD$ problem, is then constructed using the method described in Section \ref{Section:SyndromeDecodingProblems}. 

One can check that the final list created by this algorithm contains solutions to the problem. In particular, elements of top level's list are of the form $\Hm \cdot \bv$, with $wt(\bv) = N\omega_0$. That comes from the property of the weight function we use (see Equation (\ref{Constriant:WeightConstraint})) and the definitions given earlier in this subsection. An example of the algorithm for $a = 3$, $\ie$, three levels algorithm, is presented below. 

\begin{figure} \label{Figure: WagnerVar1}
	\begin{minipage}{\textwidth}
		\includegraphics[width = \textwidth]{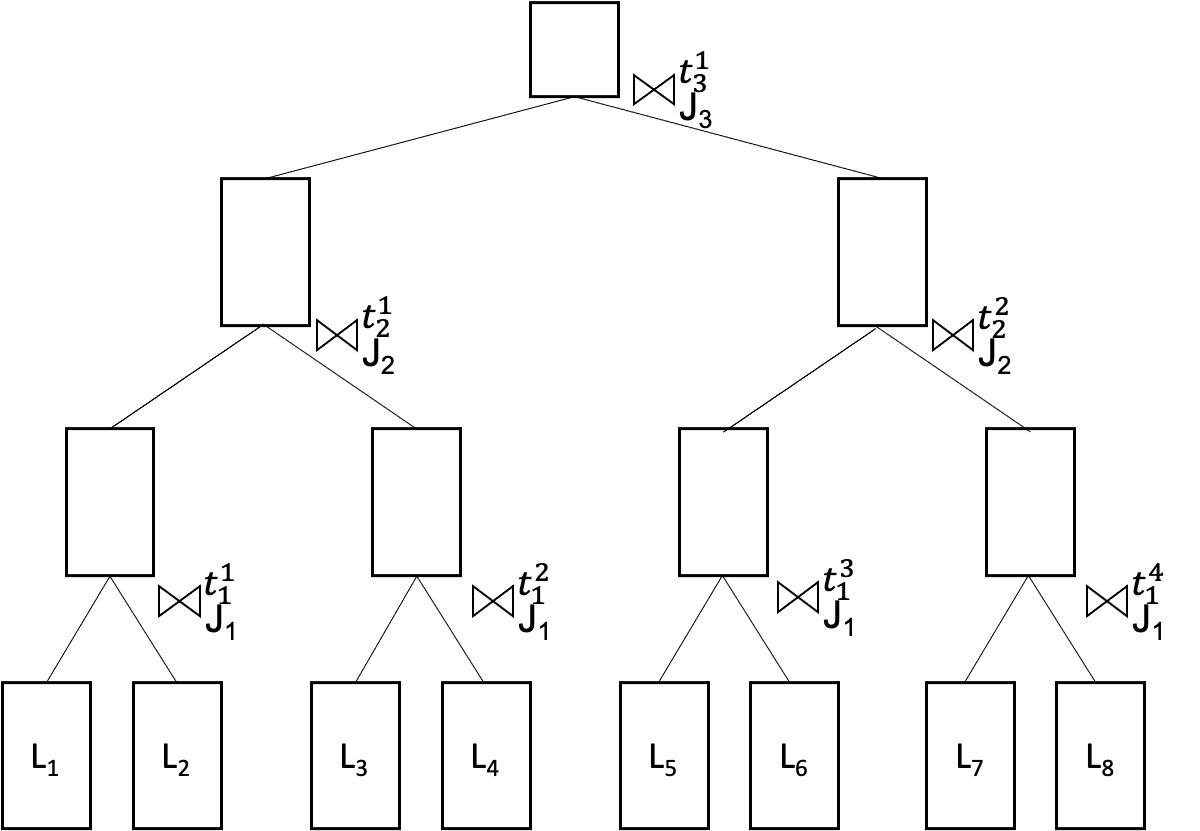}  
	\end{minipage}
	\caption{First variant of Wagner's based algorithm for $a = 3$.}
\end{figure}

\begin{proposition}\label{Proposition:Wagner1} \footnote{Notice that $Y$ and $Z$ in both propositions are determined by $m_0,\omega_{0},$ and $a$ and cannot be chosen arbitrarily.\label{proposition}}
	Fix parameters $m_0,\omega_{0}$, as well as a number of levels, $a$. 
	Let $s_{\omega_0} = \lim_{n \rightarrow \infty} \frac{1}{n}\log_q(S^n_{n\omega_{0}})$, $u = \min\{\frac{s_{\omega_0}}{2^a},m_0/a\}$, and $x = m _0- (a-1)u$.
	The first variant of the algorithm on $a$ levels solves the $\textup{\CMSD}(N,m_0N,\omega_{0}N,Y,Z )$ problem  in time $T_{\textup{CMSD}}$, where 
		\begin{align*}
			Z  = q^{N(2u - x + o(1))}, \	T_{\textup{CMSD}} = q^{N(u + o(1))}, \ Y = T_{\textup{CMSD}},
		\end{align*}
		and the $o(1)$ hides an expression that goes to $0$ as $N$ goes to $+ \infty$.
	\end{proposition}

\begin{proof}
Let us take all bottom list $L_i, \dots, L_{2^a}$, to be random subsets of size $q^{N{u}}$ of $L_i^f, \dots, L_{2^a}^f$, respecitvely
\footnote{Notice that $\lim_{n \rightarrow \infty} \frac{1}{N} \log_q|L_i^f| = \lim_{N \rightarrow \infty} \frac{1}{N} \log_qS^{N/2^a}_{N\omega_0/2^a} = \frac{1}{2^a} s_{\omega_0},$ so we can choose asymptotically any $u \le \frac{s_{\omega_0}}{2^a}$.}. Without loss of generality, we also choose $|J_j| = u$, for $j \in [2^{a-1}]$, and $|J_a| = x$. We thus have that $Y =q^{N{u}}$. Furthermore, from the merging algorithm, described earlier in this subsection, we know that all the lists up to the top level are of size $q^{Nu}$, and the list at the top level is of size $q^{N(2u - x)}$. As all the elements in the final list are solutions to the original problem, we expect $Z = q^{N(2u - x)}$ solutions, on average. All the list mergings take time $q^{Nu}$, except the last one that takes time $q^{N(2u - x)}$, hence $T_{CMSD} = \max(q^{Nu}, q^{N(2u - x)})$. From the proposition, we know that $u \leq m_0/a$ and $x = m_0 - (a-1)u$, which implies $x \geq m_0/a \geq u$, and thus $T_{CMSD} = q^{Nu}$. Therefore, we have an algorithm that finds $Z = q^{N(2u - x)}$  solutions in time $T_{CMSD} = q^{Nu}$, and for $Y = T_{CMSD} = q^{Nu}$.
\end{proof}
	
\subsection{Second Variant}

Starting from the original Wagner's algorithm \cite{Wag02}, we derive a quantum version of it and utilize it as part of an algorithm that solves the $\CMSD$ problem. Our results are presented in the rest of the section.\\

We start from relevant definitions. For a number of levels $a$, where $2^{a+1}|n$ and, for each $i \in [2^a - 1]$, we define:
\begin{multline*} 
	\I_i  \eqdef \{\bv_i \in \F_q^n : \bv_i = (\zerov^{(i-1) n/(2^a + 1)}, \widetilde{\bv_i},\zerov^{((2^a + 1)-i)  n/(2^a + 1)}) \\ \textrm{ with } \widetilde{\bv_i} \in \F_q^{n/(2^a + 1)} \wedge \wt{\widetilde{\bv_i}} = N\omega_0/(2^a + 1)\},
\end{multline*}
\begin{align*}
	L^f_i &\eqdef  \{\Hm \cdot \bv_i\}_{\bv_i \in \I_i}.
\end{align*}

For $i = 2^a$, we let:
\begin{multline*} 
	\I_{2^a}  \eqdef \{\bv_{2^a} \in \F_q^n : \bv_{2^a} = (\zerov^{(2^a-1) n/(2^a + 1)}, \widetilde{\bv_{2^a}}) \\ \textrm{ with } \widetilde{\bv_{2^a}} \in \F_q^{2n/(2^a + 1)} \wedge \wt{\widetilde{\bv_{2^a}}} = 2N\omega_0/(2^a + 1)\},
\end{multline*}
\begin{align*}
	L^f_{2^a} &\eqdef  \{\Hm \cdot \bv_{2^a}\}_{\bv_{2^a} \in \I_{2^a}},\footnotemark
\end{align*} \footnotetext{From the definitions, it can be easily seen that  $|L^f_i| = \frac{s_{\omega_0}}{2^a + 1}$, for all $i \in [2^a - 1]$, and $|L^f_{2^a}| = \frac{2s_{\omega_0}}{2^a + 1}$, for $i = a$.}

In both cases, $\Hm$ is a parity check matrix, which is defined in Section \ref{Section:InformationSetDecodingFramework}.

Like in the first variant of the algorithm, the indexing sets, $J_1, \dots, J_a$, are chosen so that they form a partition of $[n]$. The random vectors, $\tv_j^i \in \F_q^n$, for all $i \in [2^a]$ and all $j \in [a]$, also satisfy the same constraints as in the first variant (for more details, see Section \ref{Section:WagnerVar1}).

In this variant, all the bottom lists, $L_1, \dots, L_{2^a-1}$, are of the same sizes, except the rightmost one,  $L_{2^a}$, which is quadratically larger than the others. We thus change our definitions of $L^f_i$ accordingly (see definitions above). In contrast to the first variant, the algorithm does not create the rightmost list. It computes and sorts the other lists in lexicographical order on the indices of corresponding $J_j$, for all $j \in [a]$ and, instead of creating the last list, it evaluates a function that describes the list, and then finds a corresponding element (if one exists) in the top list using an efficient (quantum) routine. For the rest of the lists, the algorithm use the same merging method as in the first variant (see Section \ref{Section:WagnerVar1}). An example of the algorithm on  three levels is presented below. 

Let us now construct the function $f$ as it is required for the $\CMSD$ problem. First, let $\yv^1_{2^a},\dots,\yv^Y_{2^a}$ be the elements of $L_{2^a}$, $\ie$, the elements of the bottom right list. For a fixed $k$, we aim to find $\yv'_1,\dots,\yv'_{2^a - 1}$ that satisfy the following: for $\forall i \in [2^{a} - 1]$, $\yv'_i \in L_i$ and $\sum_i \yv'_i + \yv^k_{2^a} = \sv$. If they exist, for each $i$, we find the associated $\bv_i$ (from the definition of $\I_i$ above) such that $\Hm \bv_i = \yv'_i$ and $\Hm \bv_{2^a} = \yv^k_{2^a}$. If there are several such combinations, we take the first one according to the lexicographical order. Finally, let us take $\ev_k = \sum_i \bv_i$, so that we have $\Hm \ev_k = \sv$. We then define $f$ as follows:
$$ f(k) = \left\{\begin{array}{l}
	\ev_k \textrm{, if such a vector exists, } \\
	\zerov \textrm{, otherwise. }
\end{array}\right.$$


The function $f$ then can be described as follows. On an input $k$, $f$ takes  $\yv^k_{2^a}$, from the list $L_{2^a}$, and checks if it can be summed with $\yv'_{2^a-i}$ from the left neighbouring list, $L_{2^a-i}$, so that they appear in the solution sum. Again, if we have several such combinations, we take any one of them, for example, the first one in lexicographical order. The function repeats that at each level until it fails (in which case it outputs $\zerov$), or it arrives to the top list, where it outputs the corresponding $\ev_k$.


\begin{proposition}\label{Proposition:Wagner2}\footref{proposition}
	Fix parameters $m_0,\omega_{0}$, as well as a number of levels, $a$. 
	Let $s_{\omega_0} = \lim_{n \rightarrow \infty} \frac{1}{n}\log_q(S^n_{n\omega_{0}})$, $u' = \min\{\frac{s_{\omega_0}}{2^a + 1},m_0/a\}$, and $x = m_0 - (a-1)u$.
		The second variant of the algorithm on $a$ levels solves the $\textup{\CMSD}(N,m_0N,\omega_{0}N,Y,Z )$ problem  in time $T_{\textup{CMSD}}$, where 
		\begin{align*}
			Z = q^{N(3u' - x + o(1))}, \ T_{\textup{CMSD}} = q^{N(u' + o(1))}, \ Y = q^{N(2u' + o(1))},
		\end{align*}
		and the $o(1)$ hides an expression that goes to $0$ as $N$ goes to $+ \infty$.
\end{proposition}

\begin{proof}
We choose lists $L_1, \dots, L_{2^a-1}$ to be random subsets of size $q^{Nu'}$ of $L_1^f, \dots, L_{2^a - 1}^f$, respectively. We also choose $L_{2^a}$ so that is a random subset of $L _{2^a}^f$ and that it is of size $q^{N2u'}$. Without loss of generality, we choose $J_j$ such that $|J_j| = u'$, for all $j \in [2^{a-1}]$, and $|J_a| = x$.  We then have that $Y = q^{2Nu'}$. After the list merging at each level up to the top one, the new lists are of expected size $q^{Nu'}$, except the rightmost one, at each level, that is of expected size $q^{N2u'}$. At the top level, there is one list of the expected size $q^{Nu'}$ and one of the expected size $q^{N2u'}$. Since $|J_a| = x$, the expected size of the top list, that is  the expected number of solutions to be find by the algorithm, is $Z = q^{N(3u' - x)}$. The time for which the algorithm finds $Z$ solutions is calculated as follows.
Constructing and sorting the lists to compute $f$ take time $q^{N(u'+o(1))}$ (omitting the constant multiplicative term $2^a$), but computing $f$ afterwards take polynomial time, so we finally have $T_{CMSD} = q^{N(u'+o(1))}$. The number of $k$ such that $f(k)$ outputs a good solution is actually the size of $L^{top}$, {\ie }, $q^{N(3u' - x)}$ and, since $f : [Y] \rightarrow \F_q^n$, this proves our proposition. 
\end{proof}

\vspace*{-1cm}
\begin{figure}[H]
	\begin{minipage}{\textwidth}
		\includegraphics[width = \textwidth]{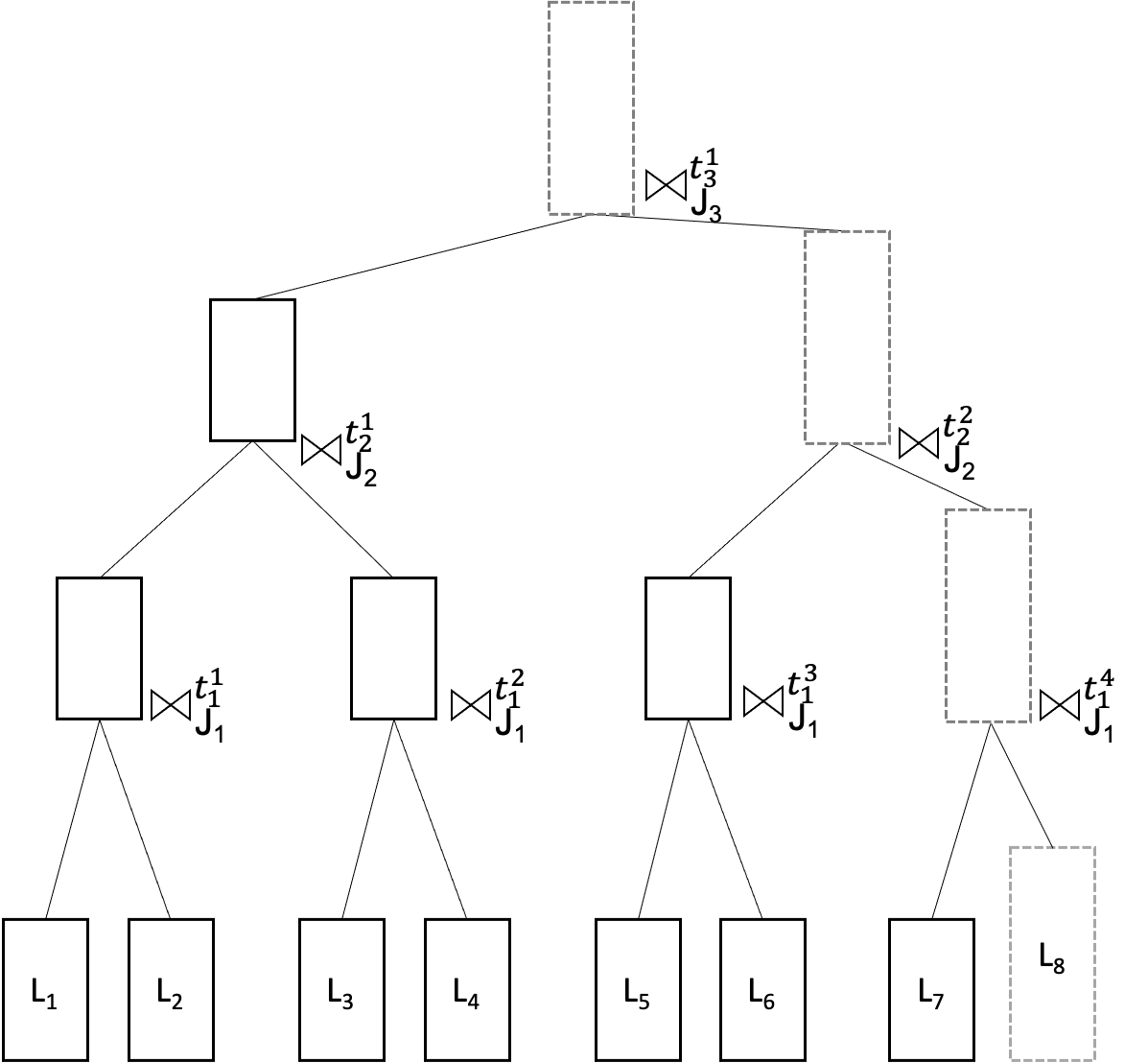}
	\end{minipage}
	\caption{Second variant of Wagner's based algorithm for $a = 3$.}
\end{figure}

\vspace*{-1cm}
\paragraph{Final remarks.}
Our ISD algorithm aims to solve an instance of $\CMSD(k+\ell,\ell,p,Y,Z)$. That means we can use the above propositions to calculate the asymptotic running time of the algorithm described in section \ref{Section:InformationSetDecodingFramework}. We first define $N = k+\ell, m_0 = \frac{\ell}{k +\ell}, $ and $\omega_{0} = \frac{p}{k+\ell},$ and then do the following: we plug Proposition \ref{Proposition:Wagner1} into Proposition \ref{Proposition:ISDC}, for the classical case, and plug Proposition \ref{Proposition:Wagner2} into Proposition \ref{Proposition:ISDQ}, for the quantum case. We then optimize parameters of our ISD algorithm over $k,\ell$, and $a$ by minimizing the algorithm's running. From the values of $k, \ell$, and $a$, we finally extract $Y$ and $Z$ and obtain the asymptotic running time of the algorithm in both the classical and quantum case.

\section{Computing Surface Area of a Sphere}

We here rely primarily on the combinatorial approach presented in \cite{A84a}. Some of the other methods are shown in more recent papers as, for example, \cite{GS91}, \cite {BB19}, \cite{WKH+21}. We decided to use the approach from \cite {A84a} as it enables us to derive a generic method for calculating the asymptotic value of the sphere surface area independently of the weight function and the alphabet size.

\begin{proposition}\label{Proposition:SphereSurfaceArea}
Fix a parameter $q$, and a weight function $wt'$ satisfying Equation \ref{Constriant:WeightConstraint}. Let the set $C$ be defined as follows: $$C \eqdef \{\cv = (c_{1},\cdots,c_{q}): i \in [q], c_i \in \mathbb{N}, \sum_{i = 1}^{q} c_i = n, \sum_{i = 1}^{q} c_i wt'(i) = w\},$$ where $w \in \mathbb{N}, \quad w \leq n \mathop{\max} \limits_{i \in \{1,\cdots, q\}} wt'(i).$ The sphere surface area, and its corresponding asymptotic value when $n$ gows to $+ \infty$, are given by the following expressions:
\begin{equation}\label{eq:volMultiSet1}
		S_{w}^{n} = \sum_{\vec{c}\in C} \binom{n}{\vec{c}}\footnote{$\binom{n}{\vec{c}}$ denotes a multinomial coefficient.}.
\end{equation}
\begin{equation}\label{eq:volMultiSetAsym}
		s_{\omega}  = \mathop{\lim}\limits_{n\to +\infty} \max_{\textbf{c} \in C} \bigg(\sum_{i=1}^{q} - \frac{c_i}{n} \log_q{\frac{c_i}{n}}\bigg).
\end{equation}
\end{proposition}

\begin{proof}

Let us first take a multiset of size $n$ where elements are taken from $[q]$, and each element is repeated $c_i$ times, for each $i \in [q]$. The number of permutations of such a multiset is given by the multinomial coefficient, defined as ${n \choose {c_{1}, \dots, c_{q}}} \eqdef \frac{n!}{c_1!\dotsc_{c_q}!}.$ This number corresponds to the number of vectors consisting of $c_1$ ones, $c_2$ twos, ..., $c_q$ values of $q$. By the definition of the set $C$, and the sphere surface area, we thus have  $S_{w}^{n} = \sum_{\vec{c}\in C} \binom{n}{\vec{c}}.$

Given the classical combinatorial result for the number of multinomial coefficients for a fixed $n$ and $q$, the size of a set $C$, and thus the number of the elements in the sum, is upper bounded by $\binom{n+q-1}{q-1}$. The upper and lower bounds of $S_{w}^{n}$ are then given by $\mathop{\max}\limits_{\vec{c}\in C} \binom{n}{\vec{c}} \leq S_{w}^{n} \leq \binom{n+q-1}{q-1}\mathop{\max}\limits_{\vec{c}\in C} \binom{n}{\vec{c}}.$

Following the same line of reasoning as in \cite{A84a}, $\ie$, by taking $log_q$ of each part of the equation above, multiplying them by $\frac{1}{n}$, where ${n\to +\infty}$, and using Stirling's approximation we finally obtain: $s_{\omega}  = \mathop{\lim}\limits_{n\to +\infty} \max_{\textbf{c} \in C} \bigg(\sum_{i=1}^{q} - \frac{c_i}{n} \log_q{\frac{c_i}{n}}\bigg).$
\end{proof}

This proposition can be observed as a generalization of the combinatorial approach presented in \cite{A84a} for any weight function and arbitrary alphabet size. Using the same reasoning, we calculate the asymptotic value of the sphere surface area, $s_{\omega}$, by reducing the Expression \ref{eq:volMultiSetAsym} to the following convex optimization problem:

\begin{problem}\label{Problem:ConvexOptProblem}
Let $\boldsymbol{\lambda} = (\lambda_1,..., \lambda_{q})$, and $\lambda_i \in \mathbb{R}_{+}$ for each $i \in [q]$. 	\begin{itemize}
		\item Maximize: $\quad -\sum_{i=1}^{q} \lambda_i \log_q{\lambda_i}$,
		\item Subject to: $\quad \sum_{i=1}^{q} \lambda_i = 1, \quad \sum_{i=1}^{q}\lambda_i wt'(i) = \omega$.
	\end{itemize}
\end{problem}

It can be easily verified that when replacing the optimization variable $\lambda_i$ with $c_i/n$  from $(5)$, the optimization problem remains convex. If we denote by $\tilde{\boldsymbol{\lambda}} = (\tilde{\lambda}_1, ..., \tilde{\lambda}_{q})$ the solution of Problem 3, the asymptotic value of the sphere surface area is calculated as $s_{\omega}  = - \sum_{i=1}^{q} \tilde{\lambda}_i \log_q{\tilde{\lambda}_i}.$ Notice here that we do not  compute only the surface areas but also the typical weight pattern of words of Lee weight $w$, {\ie } the $\cv \in C$ that maximizes the quantity in Equation \ref{eq:volMultiSetAsym}. This is necessary if we want to use this problem in Stern's signature scheme. \\

It can be shown that Problem 3 belongs to the subclass of the convex optimization problems, namely the class of conic optimization problems \cite{Boyd}. As such, it is susceptible to solving via MOSEK solver \cite{mosek}, so we utilize MOSEK as a primary computational tool. Nevertheless, to be solved via MOSEK, Problem 3 needs to be transformed so that it aligns with the standard form of conic optimization problems, as presented in the following problem: 

\begin{problem}
	Let $\vec{\lambda}\eqdef (\lambda_1,..., \lambda_{q})\in \mathbb{R}_{+}^{q}$ and $\vec{\tau} \eqdef (\tau_1,..., \tau_{q})\in \mathbb{R}_{+}^{q}$. 
	\begin{itemize}
		\item Maximize: $\quad \sum_{i=1}^{q} \tau_i$,
		\item Subject to: $\quad \sum_{i=1}^{q} \lambda_i = 1, \quad \sum_{i=1}^{q}\lambda_i \; wt'(i) =\omega, \quad (1,\vec{\lambda},\vec{\tau}) \in K_{exp}$.
	\end{itemize}
\end{problem}

where the constraint $(1,\vec{\lambda},\vec{\tau}) \in K_{exp}$ means that $\tau_i \leq - \lambda_i \log_q \lambda_i$, for each $i \in [q]$.\footnote{The notation $K_{exp}$ comes from the MOSEK optimizer\cite{mosek} and represents the exponential convex cone.} It can be easily verified that Problem 3 and Problem 4 are equivalent, hence finding a solution of either of the two yields the asymptotic value of the sphere surface area.

\section{Results}

We use our framework to compare $\SD$ with the Hamming and Lee weight. For $q = 2$ and $q = 3$, the weight functions are the same by their definitions. For $q > 3$, however, our numerical results show that the asymptotic complexities of the problem differ in these two cases and that the problem is indeed harder in the Lee weight case. We present here the comparison of the complexities of our classical ISD algorithm in the Lee and Hamming weight setting and in the parameter range that is interesting from the perspective of the hardest instances of the $\SD$ problem. It can be easily verified that the complexity of the hardest instances of the Lee $\SD$ problem is indeed higher than that of the hardest Hamming instances.

\begin{figure}[H]\label{Figure:plot2}
	\includegraphics[width=\textwidth]{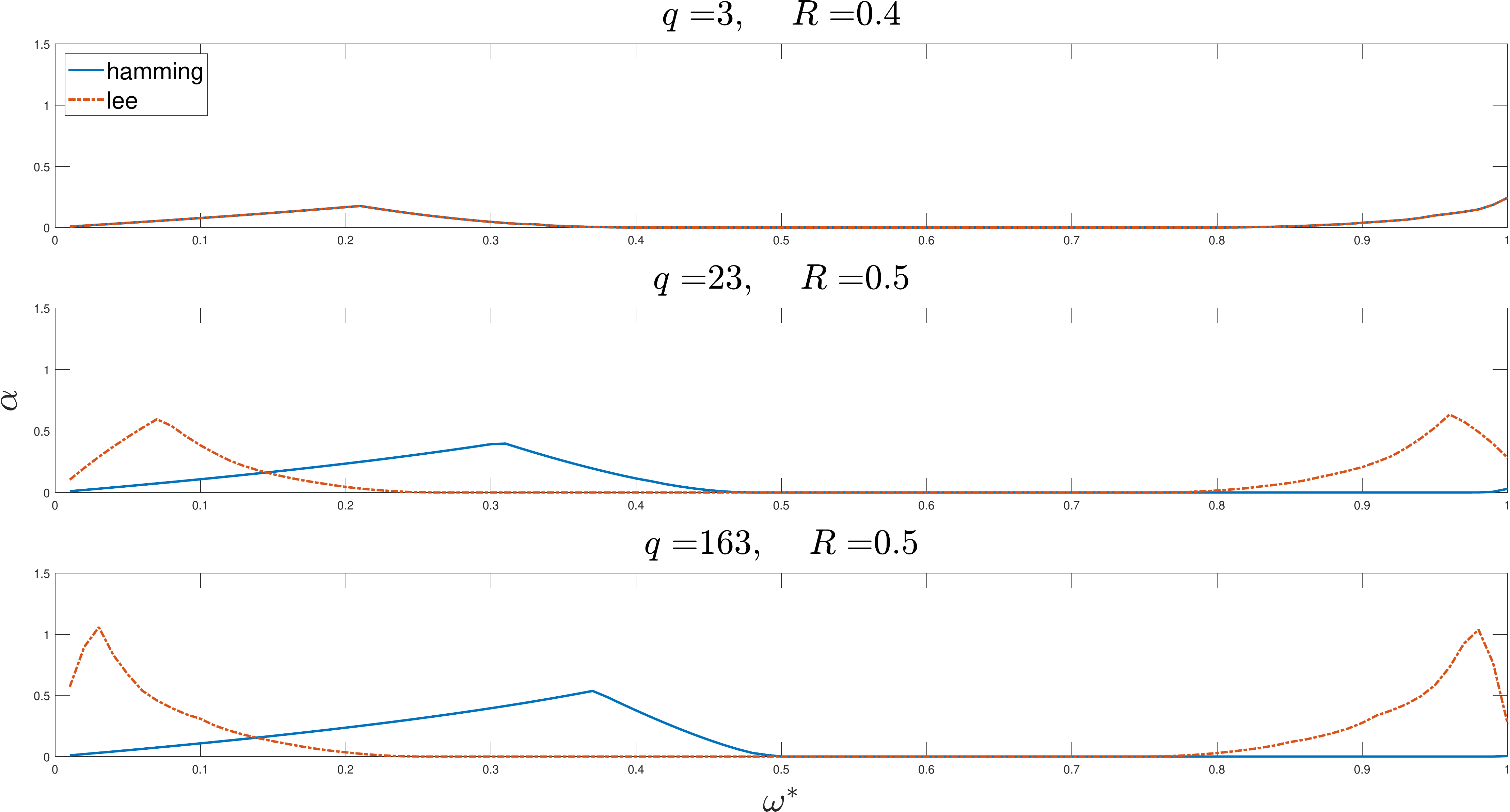}
	\caption{Comparison of the Hamming and Lee $\SD$ problem: for a fixed $q$ and $R$, the exponents $\alpha$ s.t. $Time = 2^{\alpha n}$ are given as a function of $\omega^*$, where $\omega^* = \omega$ in the Hamming weight case, and  $\omega^* = \omega \lfloor q/2 \rfloor $ in the Lee weight case.}	
\end{figure}

In the rest of the analysis, we focus on the $\SD$ problem in Lee weight. The following plot illustrates some of the numerical results we obtain.

\begin{figure}[H]\label{Figure:plot1}
	\includegraphics[width=\textwidth]{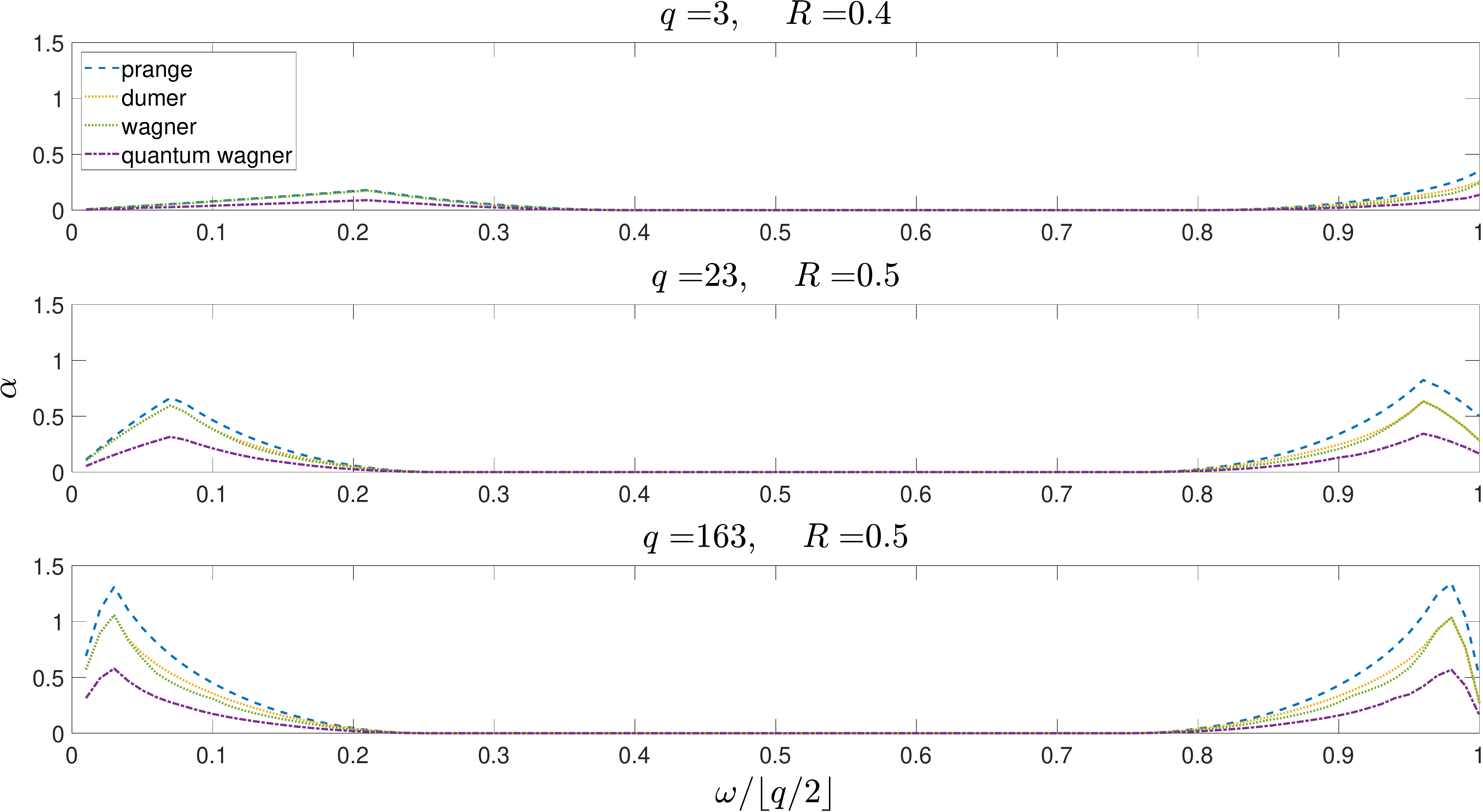}
	\caption{Hardness of the Lee $\SD$ problem: The exponents $\alpha$ of the binary asymptotic complexity, $Time = 2^{\alpha n}$, of four ISD algorithms in Lee weight setting.}
\end{figure}

We observe that for any fixed $q$ and $R$, the asymptotic complexity of our algorithms, as a function of $\omega^*$, has $2$ local maxima: at some values $\omega^*_- \in [0,{x})$ and $\omega^*_+ \in ({x},\lfloor \frac{q}{2} \rfloor]$, with $x = \frac{q^2 - 1}{4q}$\footnote{This value corresponds to the average Lee weight of a vector chosen uniformly at random.}. Moreover, these local maxima always satisfy:
\begin{align*}
	\omega^*_- & = \omega \in [0,x)\textrm{ st. } {s}_{\omega} = (1-R). \footnotemark \\
	\omega^*_+ & = \left\{\begin{tabular}{l}
		$	\omega \in (x,\lfloor \frac{q}{2} \rfloor]\textrm{ st. } {s}_{\omega}  = (1-R)  \footnotemark[\value{footnote}] \textrm{ if such an } \omega \textrm{ exists } $\\
		$	\lfloor \frac{q}{2} \rfloor \textrm{ otherwise.}$
	\end{tabular}
	\right.
\end{align*}

\footnotetext{In that case, we have $S^n_w = q^{n-k}$, which is the case where we have on average $1$ solution to the $\SD$ problem on random inputs $H,\sv$.}

This characterization of the local maxima is particularly useful when aiming to obtain the hardest instances of a problem. Namely, for a fixed $q$, it allows us to find the $R$ that yields the hardest problem and then to check only the $2$ corresponding weights, $\omega_-$ and $\omega_+$, to obtain the hardest instance. That makes our calculations more efficient, which becomes increasingly important as $q$ increases and the convex optimization part of the calculations becomes costly due to the number of constraints in Problem \ref{Problem:ConvexOptProblem}.

It is also important to notice here that many previous papers only consider the case $\omega^*_-$ and miss out on very interesting parameter ranges where, for the lower values of $q$, the problem is typically the hardest. Nevertheless, we also observe that as $q$ increases, the plots become symmetric between small weight and large weight. Therefore, we can expect that for relatively high values of $q$ the difference between the small and large weights would become negligible. However, we cannot verify this claim due to the high computational cost of such verification.

The properties we observe here hold for all ISD algorithms we consider, in both classical and quantum settings. However, it is worth noticing that while these seem to be a generic property of ISD algorithms, there might be other algorithms for which these properties do not hold.

\subsection*{Parameters for which the problem is the hardest.}

To find the hardest instances of the problem, for a given $q$, we rely on the observation about the local maxima, $\omega^*_-$ and $\omega^*_+$, and we optimize over $R$ to obtain the hardest instance. For the sake of simplicity, in Table \ref{table:HardestProblem}, we  present only the results of the analysis of the classical and quantum Wagner's based ISD algorithms and remark that the other two ISD algorithms exhibit similar behaviour.

\begin{table}[H]
\begin{center}
	\begin{tabular}{|| m{1cm} || m{1.2cm} m{1.2cm} m{1.2cm} m{1.2cm} || m{1.2cm} m{1.2cm} m{1.2cm} m{1.2cm}||}
		\hline
		\hline  
	 	 \multirow{2}{1cm} {q}& \multicolumn{4}{c ||}{Classical Wagner ISD complexity} & \multicolumn{4}{c||}{Quantum Wagner ISD complexity} \\
		\cline{2-9}
		&$R$ & $\omega/\lfloor q/2 \rfloor$ &\quad $\alpha$ & \quad $\hat{\alpha}$  &$R$ & $\omega/\lfloor q/2 \rfloor$ &\quad $\alpha$  & \quad $\hat{\alpha}$  \\
		\hline
			3 & 0.370 & 1.000 & 0.269 & 0.170 & 0.369 & 1.000 & 0.148 & 0.093 \\
			5 & 0.572 & 1.000 & 0.357 & 0.154 & 0.569 & 1.000 & 0.206 & 0.089 \\
			13 & 0.480 & 0.957 & 0.522 & 0.141 & 0.501 & 0.962 & 0.283 & 0.076 \\
			43 & 0.454 & 0.954 & 0.794 & 0.146 & 0.472 & 0.959 & 0.429 & 0.079 \\
			163 & 0.442 & 0.967 & 1.117 & 0.152 & 0.464 & 0.971 & 0.607 & 0.083 \\
			331 & 0.438 & 0.974 & 1.291 & 0.154 & 0.464 & 0.978 & 0.703 & 0.084 \\
			\hline
		\end{tabular}
\end{center}
\caption{Hardest instances of Lee $\SD$ problem: the asymptotic complexity exponents, $\alpha$ and $\hat{\alpha}$, correspond to the binary asymptotic complexity, $Time = 2^{\alpha n}$, and $q$-ary asymptotic complexity, $Time = q^{\hat{\alpha} n}$, respectively.}
\label{table:HardestProblem}
\end{table}

It can be readily verified that the complexity of a problem, expressed as $2^{n(\alpha + o(1))}$, becomes higher as $q$ increases. That is expected since the inputs’ size also increases, and we do not get this extra difficulty for free. If, for example, we want to use this problem in Stern's signature scheme, where the signature size essentially scales with the size of $q$-ary vectors of size $n$ or $n-k$, this increase of the input size becomes relevant. Therefore, we propose the scaling where the complexity is of the form $q^{n(\hat{\alpha} + o(1))}$ instead of $2^{n(\alpha + o(1))}$, and we refer to them as $q$-ary asymptotic complexity and binary asymptotic complexity, respectively. Observing $q$-ary complexity, the problem now is the hardest for $q = 3$. Intricately, $q$-ary complexity diminishes and then increases again at some point as $q$ increases. Hence, it would be interesting to calculate the asymptotic $q$-ary complexity when both $q$ and $n$ grows beyond bounds. We can also observe that while for $q = 3$ and $q = 5$ the optimal values were for $\omega^* = 1$, this property does not hold for larger $q$. Nevertheless, it remains in the range close to $1$ (typically, in the range $(0.95, 1]$). We can see, as well, that the hardest instances of the problem occur at the mid-range code rates and, typically, in the range $(0.35, 0.6)$. 

\section{Conclusion}

This paper analyzes different ISD algorithms, both in the classical and quantum regimes, for solving SD problems with varying sizes of alphabet and different weight functions. In the numerical part of the paper, we focused on analyzing the Hamming and Lee weight cases as representative examples of weight functions. 

Our results show that, for a fixed alphabet size $q > 3$, the complexity of the hardest instances of SD problem is higher in the Lee than in the Hamming weight, as well as that the hardest instances occur at high weights. That is true both in the classical and quantum setting. We also show that the problem remains exponentially hard for conveniently chosen parameters both in the classical and quantum setting for the class of the algorithms we consider. Finally, for a fixed alphabet size, we offer a rough estimate of the parameters' ranges for which the $\SD$ problem in Lee weight is typically the hardest.

These results have several implications for designers that want classical and quantum security estimates for their code-based schemes using different weight functions as, for example, for \textsf{WAVE} or other recently proposed schemes\cite{BBCHPSW2020}. For the quantum setting, our algorithms have almost a quadratic improvement over the classical setting, so it is important to update the parameters if we want to achieve quantum security. 

\vspace*{-0.3cm}

\subsubsection*{Acknowledgments.\newline \newline}

\begin{minipage}{0.5\textwidth}
	The authors want to thank Nicolas Sendrier and Anthony Leverrier for helpful discussions. S.E. has received funding from the European Union's Horizon 2020 research and innovation program under the Marie Skłodowska-Curie grant agreement No 754362.
\end{minipage}
\hspace{15mm}
\begin{minipage}{0.35\textwidth}
	    \includegraphics[width=\textwidth]{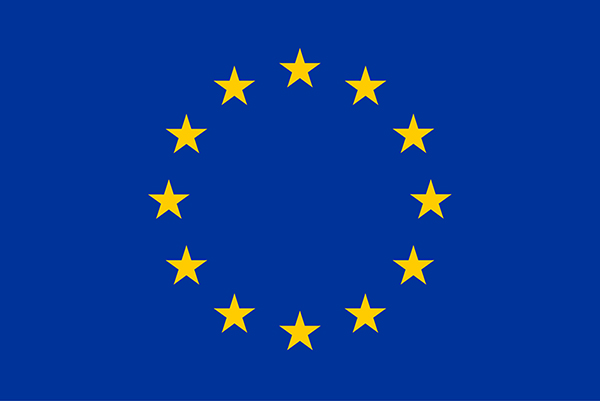}
\end{minipage}
	
\bibliographystyle{alpha}
\bibliography{extra,codecrypto}

\end{document}

%% file: settings.tex
\usepackage[utf8]{inputenc}
\usepackage[T1]{fontenc}
\usepackage{amsmath}
\usepackage{amssymb,empheq,mathtools,dsfont}
\usepackage{physics}
\usepackage{stmaryrd}
\usepackage[mathscr]{eucal}
\usepackage{color}
\usepackage{nicefrac}
\usepackage{thm-restate}
\usepackage{calrsfs}
\DeclareMathAlphabet{\pazocal}{OMS}{zplm}{m}{n}

%% file: macro_llncs.tex
\usepackage{scalerel,stackengine}
\stackMath
\newcommand\reallywidehat[1]{%
\savestack{\tmpbox}{\stretchto{%
  \scaleto{%
    \scalerel*[\widthof{\ensuremath{#1}}]{\kern.1pt\mathchar"0362\kern.1pt}%
    {\rule{0ex}{\textheight}}
  }{\textheight}%
}{2.4ex}}%
\stackon[-6.9pt]{#1}{\tmpbox}%
}

\newcommand{\F}{\mathbb{F}}

\newcommand{\Fq}{\F_q}

\renewcommand{\vec}[1]{\mathbf{#1}}

\newcommand{\bv}{\vec{b}}
\newcommand{\cv}{\vec{c}}

\renewcommand{\ev}{\vec{e}}
\newcommand{\sv}{\vec{s}}
\newcommand{\tv}{\vec{t}}

\newcommand{\xv}{\vec{x}}
\newcommand{\yv}{\vec{y}}

\newcommand{\un}{\vec{1}}

\newcommand{\Hm}{\vec{H}}

\newcommand{\Sm}{\vec{S}}

\newcommand\wt[1]{\abs{\vec{#1}}}

\makeatletter
\newcommand*{\transp}{{\mathpalette\@transpose{}}}
\newcommand*{\@transpose}[2]{\raisebox{\depth}{$\m@th#1\intercal$}}
\makeatother

\newcommand*{\eqdef}{\stackrel{\text{def}}{=}}

\newcommand{\Hmpi}{{\mathbf{H}_{\pi}}}

\newcommand{\zo}{\{0,1\}}

 \makeatletter
\newcommand\bibalias[2]{%
  \@namedef{bibali@#1}{#2}%
}

\newtoks\biba@toks
\newcommand\acite[2][]{%
  \biba@toks{\cite#1}%
  \def\biba@comma{}%
  \def\biba@all{}%
  \@for\biba@one:=#2\do{%
    \@ifundefined{bibali@\biba@one}{%
      \edef\biba@all{\biba@all\biba@comma\biba@one}%
    }{%
      \PackageInfo{bibalias}{%
        Replacing citation `\biba@one' with `\@nameuse{bibali@\biba@one}'
      }%
      \edef\biba@all{\biba@all\biba@comma\@nameuse{bibali@\biba@one}}%
    }%
    \def\biba@comma{,}%
  }%
  \edef\biba@tmp{\the\biba@toks{\biba@all}}%
  \biba@tmp
}
\makeatother

\bibalias{Ouroboros-R}{AABBBDGHZ17}
\bibalias{RQC}{AABBBDGZ17}
\bibalias{RQC2}{AABBBDGZCH19}
\bibalias{ROLLO}{ABDGHRTZABBBO19}
\bibalias{LAKE}{ABDGHRTZ17}
\bibalias{LOCKER}{ABDGHRTZ17a}
\bibalias{RankSign}{AGHRZ17}

%% file: generic_syndrome_decoding_problems.bbl
\newcommand{\etalchar}[1]{$^{#1}$}
\begin{thebibliography}{BBC{\etalchar{+}}20b}

\bibitem[ApS21]{mosek}
MOSEK ApS.
\newblock {\em MOSEK Fusion API for C++. Version Release 9.2.38.}, 2021.

\bibitem[Ast84]{A84a}
Jaakko Astola.
\newblock On the asymptotic behaviour of lee-codes.
\newblock {\em Discret. Appl. Math.}, 8(1):13--23, 1984.

\bibitem[Bar97]{B97b}
Alexander Barg.
\newblock Complexity issues in coding theory.
\newblock {\em Electronic Colloquium on Computational Complexity}, October
  1997.

\bibitem[BB19]{BB19}
S.~{Bhattacharya} and A.~{Banerjee}.
\newblock A method to find the volume of a sphere in the lee metric, and its
  applications.
\newblock In {\em 2019 IEEE International Symposium on Information Theory
  (ISIT)}, pages 872--876, 2019.

\bibitem[BBB{\etalchar{+}}20]{BBBGNRT20}
Magali Bardet, Pierre Briaud, Maxime Bros, Philippe Gaborit, Vincent Neiger,
  Olivier Ruatta, and Jean{-}Pierre Tillich.
\newblock An algebraic attack on rank metric code-based cryptosystems.
\newblock In {\em Advances in Cryptology - {EUROCRYPT} 2020}, volume 12107,
  pages 64--93. Springer, 2020.

\bibitem[BBC{\etalchar{+}}20a]{BBCHPSW2020}
Marco Baldi, Massimo Battaglioni, Franco Chiaraluce, Anna{-}Lena
  Horlemann{-}Trautmann, Edoardo Persichetti, Paolo Santini, and Violetta
  Weger.
\newblock A new path to code-based signatures via identification schemes with
  restricted errors.
\newblock {\em CoRR}, 2020.

\bibitem[BBC{\etalchar{+}}20b]{BBCGPSTV20}
Magali Bardet, Maxime Bros, Daniel Cabarcas, Philippe Gaborit, Ray~A. Perlner,
  Daniel Smith{-}Tone, Jean{-}Pierre Tillich, and Javier~A. Verbel.
\newblock Improvements of algebraic attacks for solving the rank decoding and
  minrank problems.
\newblock In {\em Advances in Cryptology - {ASIACRYPT} 2020}, volume 12491,
  pages 507--536. Springer, 2020.

\bibitem[BCDL19]{BCDL19}
R\'{e}mi Bricout, Andr\'{e} Chailloux, Thomas {Debris-Alazard}, and Matthieu
  Lequesne.
\newblock Ternary syndrome decoding with large weights.
\newblock {\em SAC 2019}, 2019.

\bibitem[BH97]{BH97}
Gilles Brassard and Peter H{\o}yer.
\newblock An exact quantum polynomial-time algorithm for simon's problem.
\newblock In {\em Fifth Israel Symposium on Theory of Computing and Systems,
  {ISTCS} 1997, Ramat-Gan, Israel, June 17-19, 1997, Proceedings}, pages
  12--23. {IEEE} Computer Society, 1997.

\bibitem[BJMM12]{BJMM12}
Anja Becker, Antoine Joux, Alexander May, and Alexander Meurer.
\newblock Decoding random binary linear codes in {$2^{n/20}$}: How {$1+1=0$}
  improves information set decoding.
\newblock In {\em Advances in Cryptology - EUROCRYPT~2012}, LNCS. Springer,
  2012.

\bibitem[BMvT78]{BMT78}
Elwyn Berlekamp, Robert McEliece, and Henk van Tilborg.
\newblock On the inherent intractability of certain coding problems.
\newblock {\em IEEE Trans. Inform. Theory}, 24(3):384--386, May 1978.

\bibitem[BV14]{Boyd}
Stephen~P. Boyd and Lieven Vandenberghe.
\newblock {\em Convex Optimization}.
\newblock Cambridge University Press, 2014.

\bibitem[Coo00]{Coo00}
Colin Cooper.
\newblock On the distribution of rank of a random matrix over a finite field.
\newblock {\em Random Struct. Algorithms}, 17:197--212, 10 2000.

\bibitem[DST19]{DST19a}
Thomas {Debris-Alazard}, Nicolas Sendrier, and Jean-Pierre Tillich.
\newblock Wave: A new family of trapdoor one-way preimage sampleable functions
  based on codes.
\newblock In {\em Advances in Cryptology - ASIACRYPT~2019}, LNCS, Kobe, Japan,
  December 2019. Springer.

\bibitem[Dum91]{D91}
Ilya Dumer.
\newblock On minimum distance decoding of linear codes.
\newblock In {\em Proc. 5th Joint Soviet-Swedish Int. Workshop Inform. Theory},
  pages 50--52, Moscow, 1991.

\bibitem[Gro96]{Gro96}
Lov~K. Grover.
\newblock A fast quantum mechanical algorithm for database search.
\newblock In Gary~L. Miller, editor, {\em Proceedings of the Twenty-Eighth
  Annual {ACM} Symposium on the Theory of Computing, Philadelphia,
  Pennsylvania, USA, May 22-24, 1996}, pages 212--219. {ACM}, 1996.

\bibitem[GS91]{GS91}
Dani{\`{e}}le Gardy and Patrick Sol{\'{e}}.
\newblock Saddle point techniques in asymptotic coding theory.
\newblock In {\em Algebraic Coding, First French-Soviet Workshop,}, volume 573,
  pages 75--81. Springer, 1991.

\bibitem[KT17]{KT17a}
Ghazal Kachigar and Jean-Pierre Tillich.
\newblock Quantum information set decoding algorithms.
\newblock In {\em Post-Quantum Cryptography~2017}, volume 10346 of {\em LNCS},
  Utrecht, The Netherlands, June 2017. Springer.

\bibitem[MMT11]{MMT11}
Alexander May, Alexander Meurer, and Enrico Thomae.
\newblock Decoding random linear codes in {$O(2^{0.054n})$}.
\newblock In {\em Advances in Cryptology - ASIACRYPT~2011}, volume 7073 of {\em
  LNCS}, pages 107--124. Springer, 2011.

\bibitem[MO15]{MO15}
Alexander May and Ilya Ozerov.
\newblock On computing nearest neighbors with applications to decoding of
  binary linear codes.
\newblock In E.~Oswald and M.~Fischlin, editors, {\em Advances in Cryptology -
  EUROCRYPT~2015}, volume 9056 of {\em LNCS}, pages 203--228. Springer, 2015.

\bibitem[NC00]{NC00}
Michael~A. Nielsen and Isaac~L. Chuang.
\newblock {\em Quantum Computation and Quantum Information}.
\newblock Cambridge University Press, 2000.

\bibitem[Pra62]{P62}
Eugene Prange.
\newblock The use of information sets in decoding cyclic codes.
\newblock {\em {IRE} Transactions on Information Theory}, 8(5):5--9, 1962.

\bibitem[Ste88]{S88}
Jacques Stern.
\newblock A method for finding codewords of small weight.
\newblock In G.~D. Cohen and J.~Wolfmann, editors, {\em Coding Theory and
  Applications}, volume 388 of {\em LNCS}, pages 106--113. Springer, 1988.

\bibitem[Wag02]{Wag02}
David~A. Wagner.
\newblock A generalized birthday problem.
\newblock In Moti Yung, editor, {\em Advances in Cryptology - {CRYPTO} 2002},
  volume 2442 of {\em Lecture Notes in Computer Science}, pages 288--303.
  Springer, 2002.

\bibitem[WKH{\etalchar{+}}21]{WKH+21}
Violetta Weger, Karan Khathuria, Anna-Lena Horlemann, Massimo Battaglioni,
  Paolo Santini, and Edoardo Persichetti.
\newblock On the hardness of the lee syndrome decoding problem.
\newblock 2021.
\newblock arXiv quant-ph 2002.12785.

\end{thebibliography}
